\documentclass[a4paper,11pt,reqno]{amsart}
\usepackage[english]{babel}
\usepackage[T1]{fontenc} 
\usepackage[utf8]{inputenc}
\allowdisplaybreaks  

\usepackage{fancyhdr} 

\usepackage{algpseudocode}
\usepackage{algorithm}
\usepackage{widetext}

\usepackage[all]{xy}

\usepackage{hyperref}
\usepackage{amsfonts}
\usepackage{amsmath}
\usepackage{amsthm}
\usepackage{amssymb}
\usepackage{mathrsfs}

\usepackage{todonotes}

\numberwithin{equation}{section}

\newtheorem{theorem}[equation]{Theorem}
\newtheorem{corollary}[equation]{Corollary}
\newtheorem{lemma}[equation]{Lemma}

\theoremstyle{definition}

\theoremstyle{remark}
\newtheorem{remark}[equation]{Remark}

\usepackage{mathtools, cuted}

\DeclareMathOperator{\mat}{\textrm{Mat}}
\DeclareMathOperator{\Ker}{\textrm{Ker}}
\DeclareMathOperator{\im}{\textrm{Im}}
\DeclareMathOperator{\zp}{\mathbb{Z}_{p}}
\DeclareMathOperator{\zpr}{\mathbb{Z}_{p^r}}
\DeclareMathOperator{\zprk}{\mathbb{Z}_{p^r}^{k}}
\DeclareMathOperator{\zprn}{\mathbb{Z}^{n}_{p^r}}

\DeclareMathOperator{\zprz}{\mathbb{Z}_{p^r}[z]}
\begin{document}
\emergencystretch 3em

\title{Decoding convolutional codes over finite rings. A linear dynamical systems approach}

\author{\'Angel Luis Mu\~noz Casta\~neda}
\address{Departamento de Matem\'aticas, Universidad de Le\'on, Le\'on, Spain}
\email{amunc@unileon.es}
\thanks{This work is part of the project TED2021-131158A-I00, funded by MCIN/ AEI/10.13039/501100011033 and by the European Union “NextGenerationEU”/PRTR}

\author{Noem\'i DeCastro-Garc\'ia}
\address{Departamento de Matem\'aticas, Universidad de Le\'on, Le\'on, Spain}
\email{ncasg@unileon.es}

\author{Miguel V. Carriegos}
\address{Departamento de Matem\'aticas, Universidad de Le\'on, Le\'on, Spain}
\email{mcarv@unileon.es}


\date{}




\begin{abstract}
Observable convolutional codes defined over \( \zpr\) with the Predictable Degree Property admit minimal  input/state/output representations that preserve structural properties under scalar restriction. We make use of this fact to present Rosenthal’s decoding algorithm for these convolutional codes. When combined with the Greferath-Vellbinger algorithm and a modified version of the Torrecillas-Lobillo-Navarro algorithm, the decoding problem of convolutional codes over \( \zpr\) reduces to selecting two decoding algorithms for linear block codes over a field. Finally, we analyze both the theoretical and practical error-correction capabilities of the combined algorithm as well as its time complexity.
\end{abstract}

\maketitle



\section{Introduction}\label{sec1}
P. Elias introduced convolutional codes in 1955 \cite{elias}, and G. D. Forney provided the first algebraic–theoretical approach to convolutional codes in \cite{forney}. Within the theory of convolutional codes, two fundamental problems naturally arise. The first concerns the design of convolutional codes with strong error-correction capabilities, while the second addresses the decoding problem, that is, the recovery of the transmitted information from a corrupted received word. If we focus on the last one, the search for efficient decoding algorithms for convolutional codes is a nontrivial task. For this reason, a current line of research in convolutional coding theory focuses on constructing codes with good decoding properties, commonly referred to as the \emph{Good Decodable Property} (GDP) \cite{constructionnuevo}, that is, codes for which efficient decoding algorithms exist. One effective way to address both the design and decoding problems is to study
convolutional codes through alternative mathematical frameworks. Among the most
widely studied approaches is the connection between convolutional codes and linear
dynamical systems via the well-known input/state/output (I/S/O) representations,
introduced and developed in~\cite{behaviors,bch}. These representations associate a
convolutional code (or its encoder) with a linear dynamical system. 

An I/S/O
representation of an $(n,k,\delta$)- convolutional code $\mathcal{C}$ over a finite field $\mathbb{F}_{q}$ denoted by $ \mathbb{F}$ is described by the equations
\[
x_{t+1} = A x_t + B u_t,
\]
\[
y_t = C x_t + D u_t,
\]
where $x_t \in \mathbb{F}^\delta$ denotes the state vector, $u_t \in \mathbb{F}^k$ is the input vector,
and $y_t \in \mathbb{F}^{n-k}$ represents the output vector. The collection of matrices $\Sigma (\mathcal{C})=(A,B,C,D)\in 
\mat_{\delta \times \delta}(\mathbb{F})
\times\mat_{\delta \times k}(\mathbb{F})
\times\mat_{(n-k) \times \delta}(\mathbb{F}) 
\times\mat_{(n-k) \times k}(\mathbb{F})
$
constitutes the I/S/O representation of the convolutional code $\mathcal{C}$. This framework enables the construction of decoding algorithms by explicitly
describing the algebraic dependency between the information vector and the output
vector through the state-space equations. As a consequence, the decoding problem can
be translated into the problem of solving systems of linear equations over the
underlying field. 

The first decoding algorithm based on I/S/O representations, and computationally more
efficient than the Viterbi algorithm for certain families of convolutional codes, was proposed
in~\cite{decodingros} and is valid for non-catastrophic convolutional codes. Motivated by this result, the decoding problem
of convolutional codes using I/S/O representations has been extensively studied in
the literature and remains an active area of research (see, for instance,
\cite{isabeldecoding,paconuevo,angel2,tesisvirtudes,lieb}). However, most decoding algorithms based on I/S/O representations have been developed over
finite fields. Extending these techniques to convolutional codes
defined over finite rings is a natural and relevant step, both from a theoretical
perspective and in view of potential applications.

The study of convolutional codes over rings, and in particular over
$\zpr$, was initiated by J.~L.~Massey and T.~Mittelholzer
in~\cite{Massey2,Massey}. The motivation for considering finite rings in this setting
stems from their potential applications in phase modulation schemes, where the richer
algebraic structure of rings may provide practical advantages. Over the past decades,
a substantial body of literature has been developed on convolutional codes over rings
(see  \cite{Fagnani,ccrings3,PINTOKUIJPOLDER,ccringsminimality2,mit,napp2018,OBSERVABLEFINITERINGS}),
covering both theoretical foundations and applied aspects. Despite this progress, the
extension of well-established decoding strategies from the field case to the ring
setting remains relatively underexplored, mainly from the I/S/O representations point of view.

In this article, we focus on the decoding process for convolutional codes defined over
finite rings~\cite{listdecoding1,Lstdecoding2}. In particular, we investigate how
minimal I/S/O representations can be employed to decode convolutional codes when the underlying alphabet is a finite ring such as $\zpr$.
More precisely, we study whether Rosenthal's decoding algorithm for convolutional
codes over finite fields can be extended to the ring setting, specifically to
$\zpr$. To this end, we model convolutional codes as submodules of
$\zpr[z]^n$ and recall the necessary conditions on the encoder, namely,
the Predictable Degree Property (PDP), and then an observable convolutional code, that guarantee the existence of an I/S/O
representation~\cite{isoringmunoz,NAPPfiniterings}. Under these assumptions, we reduce
the decoding problem for convolutional codes over $\zpr$ to the decoding
of two linear block codes over the same ring, in a manner analogous to Rosenthal's
decoding algorithm for observable convolutional codes over finite fields. Finally, we
provide a performance evaluation of the proposed decoding algorithm. 
Extending decoding methods for convolutional codes from finite fields to finite rings is not always a purely formal step. While Rosenthal's decoding procedure can be transferred to rings such as $\zpr$ with essentially the same algorithmic structure, several field-based arguments  require additional care in the ring setting. In particular, when using I/S/O representations, one must ensure that notions such as minimality and observability are formulated and verified appropriately over $\zpr$, and the linear systems arising in decoding may involve ring-specific subtleties. Hence, even though the decoding steps may translate with little difficulty, the supporting algebraic justifications should be revisited to avoid implicitly relying on properties that hold only over fields.

The article is organized as follows. Section II includes an overview of convolutional codes and their relation with linear dynamical systems over \( \zpr \).  The main results are developed in Section III, and they are completed in Section IV with the study of the performance of the proposed decoding algorithm. Finally, the conclusions and the references are given.


%
%

\section{Preliminaries}\label{sec:preliminaries}

We begin by introducing notation that will be used throughout the article.

We fix a prime number $p\geq 2$ and a natural number $r\geq 1$. We denote by $\zpr$ the quotient ring $\mathbb{Z}/{p^r\mathbb{Z}}$   and by $\zpr[z]$ the ring of polynomials with coefficients in $\zpr$. The residue field of the finite ring $\zpr$ is the finite field $\mathbb{F}_{p}$. 

Given an element $z\in\zpr$ with $p$-adic expansion $z=\sum_{j=0}^{r-1}z_{j}p^{j}$ with $z_{j} \in \mathbb{F}_{p}$, we define $z^{(i)}:=\sum_{j=0}^{i}z_{j}p^{j}$ and $^{(i)}z:=\sum_{j=i}^{r-1}z_{j}p^{j}$. We extend these definitions to vectors and matrices with coefficients in $\zpr$. If an element $z\in\zpr$ can be written as $z=z' p^{r-1}$ with $z'\in\mathbb{F}_{p}$, we will use the notation $z'=\dfrac{z}{p^{r-1}}$.

Finally, given an ordered set $c=(c_0,\hdots, c_{l-1})$ and natural numbers $i<j\leq l-1$, we will denote by $c[i,j]$ the ordered set $(c_i,\hdots, c_{j})$.

\subsection{Linear block codes over $\zpr$}\label{sec:block}

An $(n,k)$ linear block code over $\zpr$ is a submodule $\mathcal{C}\subset\zprn$ of rank $k$ whose elements are called codewords. Linear block codes over $\zpr$ are not necessary free submodules. Thus, we say that the code $\mathcal{C}\subset\zprn$ is free if it is a free submodule. On the other hand, the linear block code $\mathcal{C}\subset\zprn$ is said to be splitting if it is a direct summand of $\zprn$.  

An encoder $G$ for the linear block code $\mathcal{C}\subset\zprn$ is a full rank (injective) generator matrix $G\in \text{Mat}_{n \times k}(\zpr)$ whose column space is equal to $\mathcal{C}$. The encoder $G$ is said to be strong if it has a full size minor which is a unit in $\zpr$. This is equivalent to say that the $p$-adic expansion $G=\sum_{j=0}^{r-1}G_{j}p^{j}$ satisfy that $G_0$ is full rank. It is well-known that a linear code $\mathcal{C}\subset\zprn$ admits a strong encoder if and only if $\mathcal{C}$ is splitting (\cite[Lemma 2]{lincodesrings}).

Given an $(n,k)$ linear block code $\mathcal{C} \subset \zprn$, a {parity-check matrix} is a full-rank matrix $H \in \mathrm{Mat}_{(n-k) \times n}(\zpr)$ such that $Ker(H) = \mathcal{C}$. The code $\mathcal{C}$ admits a parity-check matrix if and only if it is splitting. In particular, as established in the previous paragraph, this is also equivalent to the existence of a strong encoder. Hence, the existence of strong encoders and parity-check matrices are both equivalent to the splitting property of the code.

The (minimum) distance of a linear block code  $\mathcal{C}\subset\zprn$ is defined as 
$
d_{\textrm{min}}:=\textrm{min}_{0\neq c\in\mathcal{C}}\{\omega_{H}(c)\},
$
where $\omega_{H}(c)$ is the  Hamming weight of $c\in\mathcal{C}$. As for any block code, $\mathcal{C}\subset\zprn$ can correct all weight errors up to $t$ if and only if 
\begin{equation}\label{eq:dmin}
d_{\textrm{min}}=2t + 1.
\end{equation}
An important feature of any splitting linear block code $\mathcal{C}\subset\zprn$, is that the minimum distance coincides with the minimum distance of its restriction modulo $p$ ,
$$d_{\textrm{min}}(\mathcal{C})=d_{\textrm{min}}(\mathcal{C}\textrm{ mod }p).$$

Given an $(n,k)$ splitting linear block code over $\zpr$ of rank $k$, $\mathcal{C}\subset\zprn$, a parity-check matrix $H\in \text{Mat}_{n-k \times n}(\zpr)$ and a vector  $e\in\zprn$, its syndrome is $He\in\mathbb{Z}^{n-k}_{p^r}$ while its coset is $e+\mathcal{C}\subset \zprn$. The (Hamming) weight of the coset $e+\mathcal{C}$ is $\omega_{H}(e+\mathcal{C}):=\textrm{min}_{e'\in e+\mathcal{C}}\{\omega_{H}(e')\}$. A coset leader of $e+\mathcal{C}$ is an element $e'\in e+\mathcal{C}$ such that $\omega_{H}(e')=\omega_{H}(e+\mathcal{C})$. It is clear that if $e\in\zprn$ has weight $\omega_{H}(e)\leq \lfloor(d-1)/2\rfloor$ then $e$ is the unique coset leader in its coset.

\subsection{Convolutional Codes over $\zpr$}\label{sec:conv}

Let $k$ and $n$ be natural numbers such that $k \leq n$. An $(n,k)$ convolutional code over $\zpr$ is defined as a submodule 
$\mathcal{C} \subset \zpr[z]^n$
of rank $k$ \cite[Definition 1]{isoringmunoz}. As in the linear block case, elements of $\mathcal{C}$ are called codewords. 

For an $(n,k)$ convolutional code $\mathcal{C}$ over $\zpr$, a convolutional encoder is a full rank (i.e. injective) matrix 
$G(z) \in \text{Mat}_{n \times k}(\zpr[z])$  
such that the column space of $G(z)$ is equal to $\mathcal{C}$. Note that convolutional codes over $\zpr$ admitting convolutional
encoders are necessarily free.


One important property of a convolutional code is that it is non-catastrophic; that is, observable. An $(n,k)$ convolutional code $\mathcal{C} \subset \zpr[z]^n$ is said to be {observable} if the natural quotient map 
$$\zpr[z]^n \twoheadrightarrow \zpr[z]^n/\mathcal{C}$$ 
has a right inverse, which means that $\zpr[z]^n/\mathcal{C}$ is free of rank $n-k$ \cite[Remark 1]{isoringmunoz} and  \cite[\S IIIA]{Fagnani}. In fact $\mathcal{C}$ is non-catastrophic if and only if one has the decomposition $\zpr[z]\cong\mathcal{C}\oplus\left(\zpr[z]^n/\mathcal{C}\right)$. In particular, if $\mathcal{C}$ is non-catastrophic or observable then it is a free module. From now on, every convolutional code will be assumed to be a free module.

The $i$-th constraint length, $\nu_i$, of a convolutional encoder $G(z)$ is defined as the maximum degree of the entries in the $i$-th column of $G(z)$. Without loss of generality, we may assume that $\nu_1 \geq \dots \geq \nu_k$  \cite[Definition 5]{isoringmunoz}. The maximum constraint length, $\nu_1$, is called the memory of the encoder. The complexity (or degree) of $G(z)$ is given by $\delta_{G(z)} := \sum_{i=1}^k \nu_i$. The complexity  of $\mathcal{C}$, denoted by $\delta_{\mathcal{C}}$, is defined as 
$$\delta_{\mathcal{C}}:=\textrm{min}\{\delta_{G(z)}| \ G(z) \textrm{ encoder of }\mathcal{C}\}$$
Clearly,
$$
\delta_{\mathcal{C}}\geq \textrm{max}
\left\{
\begin{array}{l|l}
\textrm{maximum degree} &\\ 
\textrm{of the full-size} & G(z) \textrm{ encoder of }\mathcal{C}\\ 
\textrm{minors of }G(z)  &
\end{array}
\right\}=:\delta'_{\mathcal{C}}
$$
As mentioned in the Introduction, an important property in order to obtain minimal I/S/O representations of a convolutional code over $\zpr$ is the PDP. We recall the definition. Given a convolutional encoder $G(z)$ of an $(n,k)$ convolutional code $\mathcal{C}$ over $\zpr$, the following holds:
$$\text{deg}(G(z)u(z)) \leq \max \{\text{deg}(u_i(z)) + \text{deg}(g_i(z))\},$$ 
where the degree of a polynomial vector is the maximum degree among its components, and $g_i(z)$ refers to the $i$-th column of $G(z)$. The encoder $G(z) \in \text{Mat}_{n \times k}(\zpr[z])$ has the Predictable Degree Property (PDP) if the above inequality becomes an equality for every $u(z)$. The PDP is characterized by the condition that the matrix of column-wise maximum degree coefficients,
$G_h$, is injective \cite[Theorem 1]{isoringmunoz}. It follows easily that if $\mathcal{C}$ is a convolutional code $\mathcal{C}$ over $\zpr$ that has a convolutional encoder $G(z)$ with the PDP then
$
\delta_{G(z)}=\delta_{\mathcal{C}}= \delta'_{\mathcal{C}}.
$
We say that a convolutional code is PDP if it admits an encoder with the PDP.

Finally, the free distance of a convolutional code $\mathcal{C} \subset \zpr[z]^n$ is defined as
$$
d_{\textrm{free}}(\mathcal{C}) = \textrm{min}\{\omega_{H}(v(z))  | \  v(z) \in \mathcal{C}, v(z)\neq 0\}.
$$
where for $v(z) = \sum_{i\in \mathbb{N}_0} v_i z^i, \omega_{H}(v(z)) = \sum_{i\in \mathbb{N}_0} \omega_{H}(v_i)$,  $\omega_{H}(v_i)$ being the Hamming weight of $w_i$ \cite[\S 1]{solering}.

\subsection{I/S/O Representations}\label{sec:iso}

Given an $(n,k)$ observable convolutional code $\mathcal{C} \subset \zpr[z]^n$ with PDP of complexity $\delta$, there exist matrices $K \in \mat_{(\delta+n-k)\times \delta}(\zpr)$, $L \in \mat_{(\delta+n-k)\times \delta}(\zpr)$, and $M \in \mat_{(\delta+n-k)\times n}(\zpr)$ which satisfy the following conditions \cite[Theorem 2]{isoringmunoz}:
\begin{equation}\label{representation}
\begin{split}
\mathcal{C} = \{v(z) \in \zpr[z]^n :& \ \exists x(z) \in \zpr[z]^\delta \text{ such that }\\
&z K x(z) + L x(z) + M v(z) = 0\},
\end{split}
\end{equation}
and such that
\begin{enumerate}
\item $K$ is injective.
\item $(K,M)$ is surjective.
\item $(zK+L,M)$ is surjective. 
\end{enumerate}
This triple is called a minimal first-order representation of the convolutional code $\mathcal{C}$ and it is unique up to similarity transformations \cite{isoringmunoz}. Moreover, there is always a permutation $\sigma \in \mathfrak{S}_n$ such that any minimal first-order representation $(K, L, M)$ of the code $\sigma(\mathcal{C})$ is similar to a minimal first-order representation of the form
\begin{equation}\label{iso}
\left( \left(\begin{array}{c} -I\\ 0 \end{array} \right), \left(\begin{array}{c} A\\ C \end{array} \right), \left(\begin{array}{cc} 0 & B\\ -I & D\end{array} \right) \right).
\end{equation}
Here, the matrices $A \in \mat_{\delta \times \delta}(\zpr)$, $B \in \mat_{\delta \times k}(\zpr)$, $C \in \mat_{(n-k) \times \delta} (\zpr)$, and $D \in \mat_{(n-k) \times k}(\zpr)$ define a reachable and observable linear dynamical system   \cite[Corollary 2, Theorem 3 and Theorem 4]{isoringmunoz}, known as the minimal I/S/O representation of the convolutional code, which describes the encoding process for $\mathcal{C}$. To see this, denote
$x(z)=x_{0}z^{\gamma}+x_{1}z^{\gamma-1}+\hdots+x_{\gamma},
\ u(z)=u_{0}z^{\gamma}+u_{1}z^{\gamma-1}+\hdots+u_{\gamma} \textrm{ and }
\ y(z)=y_{0}z^{\gamma}+y_{1}z^{\gamma-1}+\hdots+y_{\gamma}.$
Then, the Eq. \eqref{representation} for $K$, $L$, $M$ as in Eq. \eqref{iso} leads to the following system of equations:
\begin{equation}\label{eqslinearsys}
\Sigma:=\left\{
\begin{array}{rl}
x_{t+1} &= A x_t + B u_t,\\
y_t &= C x_t + D u_t,\\
v_t &= \left( \begin{array}{c} y_t\\ u_t \end{array} \right), \ x_0 = 0, \ x_{\gamma+1} = 0.
\end{array}
\right.
\end{equation}
Again, a minimal I/S/O representation of an observable convolutional code with PDP
is unique up to similarity transformations.
Conversely, let $\Sigma=(A,B,C,D)$ be a reachable and observable linear system over
$\zpr[z]^n[z]$. Then the convolutional code $ \mathcal{C} =Ker(zK+L \mid M) $ associated with $\Sigma$ is observable.
This establishes that the correspondence between reachable and observable linear
systems and observable convolutional codes with PDP holds in both directions;
see \cite[Theorem~5]{isoringmunoz}.

Associated to a minimal I/S/O representation of an observable convolutional code with PDP there are two important scalar matrices:
\begin{enumerate}
\item The reachability matrix:
\begin{equation*}\label{phi}
\Phi_{l}(\Sigma):=\left(A^{l-1}B, \hdots, AB, B\right),
\end{equation*}
which is surjective in case $l=\delta$. That is equivalent to the dynamical system $\Sigma=(A,B,C,D)$ being reachable.
\item The observability matrix:
\begin{equation*}\label{psi}
\Psi_{l}(\Sigma):=\left(
\begin{array}{c}
C\\
CA\\
\vdots\\
CA^{l-1}
\end{array}
\right),
\end{equation*}
which is injective in case $l=\delta$. That is equivalent to the dynamical system $\Sigma=(A,B,C,D)$ being observable.
\end{enumerate}

\begin{remark}
For any pair of values $T,\Theta$ making $\Phi_{T}(\Sigma)$ and $\Psi_{\Theta}(\Sigma)$ surjective and injective, respectively, we may regard these matrices as parity-check and generator matrices, respectively, of the corresponding linear block codes over $\zpr$, $\Ker(\Phi_{T}(\Sigma)), \ \im(\Psi_{\Theta}(\Sigma))$.
Rosenthal's algorithm \cite{decodingros} works under the assumption of the existence of such values $T,\Theta$ and reduces the decoding process of the convolutional code $\mathcal{C}$ to the decoding of the linear block codes $\Ker(\Phi_{T}(\Sigma)), \ \im(\Psi_{\Theta}(\Sigma))$ defined over $\zpr$. 
Therefore, linear block codes over $\zpr$ and their decoding algorithms will play an important role in Rosenthal's decoding procedure.
    
\end{remark}

We conclude this section by establishing a connection between the PDP property and the splitting property of the associated linear block code  $\im(\Psi_{l}(\Sigma))$.

\begin{lemma}
Let $\mathcal{C} \subset \zpr[z]^n$ be a $(n,k)$ observable convolutional code  of complexity $\delta$ with PDP and $\Sigma=(A,B,C,D)$  an I/S/O representation of $\mathcal{C}$. Then, for every $l\geq \delta$, $\im(\Psi_{l}(\Sigma))\subset \zpr^{l(n-k)}$ is a splitting linear code over $\zpr$ of dimension $\delta$.
\end{lemma}
\begin{proof}
Since $\mathcal{C}$ has the PDP, by \cite[Theorem 5 and Remark 7]{isoringmunoz}, the matrices $\Sigma:=(A,B,C,D)$ form a minimal I/S/O representation of $\mathcal{C}$, and their reduction modulo $p$ gives rise to matrices $\Sigma_0:=(A_0, B_0, C_0, D_0)$ that define a minimal I/S/O representation of the reduced code $\mathcal{C}_0 := \mathcal{C} \mod p$, which is a convolutional code over the finite field $\mathbb{F}_p$ of the same parameters $(n,k,\delta)$.
Let us denote by $\Psi_{l}(\Sigma_0)$ the reduction modulo $p$ of the matrix $\Psi_l(\Sigma)$. That is,
\[
\Psi_{l}(\Sigma_0) := \Psi_l(\Sigma) \mod p = 
\begin{pmatrix}
C_0 \\
C_0 A_0 \\
\vdots \\
C_0 A_0^{l-1}
\end{pmatrix}.
\]
Since the original code $\mathcal{C}$ is observable with PDP, its reduction $\mathcal{C}_0$ is also observable (see \cite[Remark 7]{isoringmunoz}). Therefore, the observability matrix $\Psi_{l}(\Sigma_0)$ is injective for all $l \geq \delta$. This means that $\Psi_l(\Sigma)$ is a strong encoder. Therefore, $\operatorname{Im}(\Psi_l(\Sigma)) \subset \zpr^{l(n-k)}$ is a splitting linear block code of rank $\delta$.
\end{proof}

\section{The decoding algorithm}

Given an $(n,k)$ convolutional code $\mathcal{C} \subset \mathbb{Z}_{p^r}[z]^n$ and a received word $\widehat{v}(z) \in \mathbb{Z}_{p^r}[z]^n$, the decoding problem consists of finding a codeword $v(z) \in \mathcal{C}$ that minimizes the Hamming weight of the error:
\begin{equation}\label{eq:decoding}
\min_{v(z) \in \mathcal{C}} \, \omega_H\big(v(z) - \widehat{v}(z)\big).
\end{equation}
If the free distance of $\mathcal{C}$ is $d$ and the error vector $e(z) := \widehat{v}(z) - v(z)$ has Hamming weight less than or equal to $\lfloor (d - 1)/2 \rfloor$, then the solution to \eqref{eq:decoding} is unique and corresponds to the originally transmitted codeword (see \cite[Proposition 2.5]{pretzel}). In this case, the decoding is said to be complete.

Note that the codewords $v(z) \in \mathcal{C}$ represent the entire information content of the transmitted message. Therefore, decoding with convolutional codes is typically performed after the complete message has been received, or once a sufficiently large portion is available to enable sliding-window decoding. This differs conceptually from the case of linear block codes $\mathcal{C} \subset \mathbb{Z}_{p^r}^n$, where the message is encoded and transmitted block by block, and decoding can be carried out independently on each block by solving
\[
\min_{v \in \mathcal{C}} \, \omega_H\big(v - \widehat{v}\big)
\]
without needing access to the entire message.

Rosenthal's decoding algorithm for convolutional codes over finite fields is based on decoding two associated linear block codes \cite{decodingros}. This structure allows decoding to be performed over segments of bounded length, provided the number of errors in each segment remains within the correctable range. The same principle extends naturally to convolutional codes over $\mathbb{Z}_{p^r}$.

In this section, we present a combined decoding procedure that merges Rosenthal's algorithm with the Greferath–Vellbinger decoding method \cite{lincodesrings} and a modified version of the Torrecillas–Lobillo–Navarro algorithm \cite{lincodesringsH}.

%
%
%
%
%
%
%

\subsection{Greferath-Vellbinger (GV) algorithm for splitting linear block codes over $\zpr$}

Let $\mathcal{C}\subset\zprn$ be an $(n,k)$ splitting linear block code and $G\in\textrm{Mat}_{n\times k}(\zpr)$ an encoder for $\mathcal{C}$. Suppose that the transmitted information word is $u \in\zprk$, the codeword is $v=Gu\in\zprn$ and that the received codeword is $\widehat{v} \in\zprn$. 
Let $e:=\widehat{v}- v$ be the error vector. 
Assume we have a decoding algorithm $\partial$ for the linear code $\mathcal{C}_{0}:=\mathcal{C} \ \textrm{mod}(p)$.

At each iteration $i$, the GV algorithm  (see \cite[pp. 1290]{lincodesrings}) computes an auxiliary quantity $\delta_i$ which isolates the contribution of the $i$-th $p$-adic digit of the information and error words. This is achieved by subtracting from the truncated received vector $\widehat{v}^{(i)}$ the contribution of previously computed digits $u^{(i-1)}$ and $e^{(i-1)}$, and dividing the resulting difference by $p^i$ (see notation at Section \ref{sec:preliminaries}):
\[
\delta_i := \frac{\widehat{v}^{(i)} - G^{(i)} u^{(i-1)} - e^{(i-1)}}{p^i}.
\]
At this point, a decoder $\partial$ (over $\mathbb{F}_{p}$) is used to solve the equation 
$$\delta_i = G_0 u_i + e_i.$$ Since $G$ is assumed to be a strong generator matrix, the code $\mathcal{C}_0$ is a linear code over $\mathbb{F}_p$ with the same {\color{red}{dimension}}. The correctness of the method relies on the uniqueness of the $p$-adic expansion: as long as each error component $e_i$ lies within the correction capacity of $\partial$, the original message can be fully reconstructed layer by layer. This approach reduces the decoding over $\zpr$ to successive decoding problems over $\mathbb{F}_p$, preserving both correctness and efficiency.

The pseudocode of the algorithm is as follows:

\begin{algorithm}[H]
\caption{GV algorithm}\label{alg:G}
\begin{algorithmic}
\State \hspace{-0.4cm} \textbf{Parameters:} $G\in\textrm{Mat}_{n\times k}(\zpr)$ encoder of an $(n,k)$ free linear block code and a decoding algorithm $\partial$ for the linear code given by $G_0:=G\ \textrm{mod}(p)$.
\State \hspace{-0.4cm} \textbf{Input:}  A received word $\widehat{v}=Gu+ e\in\zprn$.
\State \hspace{-0.4cm}  \textbf{Output:} $u,v,e\in\zprn$
\State $i=0$
\While{$i \leq r-1$} 
    \State $\delta_i \gets \dfrac{\widehat{v}^{(i)}- G^{(i)}u^{(i-1)}- e^{(i-1)}}{p^{i}}$
    \State Solve $\delta_i=G_0 u_i +e_i$ with $\partial$  
\State $i=i+1$
\EndWhile
\State \Return{$e,u, v=Gu$}
\end{algorithmic}
\end{algorithm}

\subsection{Torrecillas-Lobillo-Navarro (TLN) algorithm adapted to free linear block codes over $\zpr$}

Let $\mathcal{C}\subset\zprn$ be an $(n,k)$ linear block code and $H\in\textrm{Mat}_{q\times n}(\zpr)$ a (not necessary full-rank) parity-check matrix for $\mathcal{C}$. Suppose that the transmitted codeword is $v\in\zprn$ and that the received codeword is $\widehat{v} \in\zprn$. Let $e:=\widehat{v}- v$ be the error vector. Let $s:=H\widehat{v} =He$ be the syndrome of the received word. The TLN algorithm computes the coefficients of the $p$-adic expansion of the error $e$ iteratively. It assumes that the parity-check matrix $H$ is in \emph{standard form}, meaning that there are matrices $^{(i)}\widehat{H}\in\textrm{Mat}_{q_i\times n}(\zpr)$, $i=1,\hdots,r-1$ with $p$-adic expansion $^{(i)}\widehat{H}=\sum_{j=i}^{r-1}\Lambda_{ij}p^{j}$ such that $H=(^{(0)}\widehat{H}|^{(1)}\widehat{H}|\cdots|^{(r-1)}\widehat{H})$ (thus $\sum q_i =q$), and satisfying $\Lambda_{ii}$ is full rank for each $i$. Further, the algorithm assumes that $r$ decoding algorithms have been fixed, one for each linear code $\Ker$($\Lambda_{ii})$.

Note that condition $\Lambda_{ii}$ is full rank for each $i$ makes $^{(i)}\widehat{H}\neq 0$ for each $i$ so that $q$ must satisfy $q\geq r$.  In our case, we will have a full rank parity-check matrix $H$. We will show now how to get advantage of this fact to apply TLN algorithm without computing the standard form, and using only one decoding algorithm for the code defined by $H_0:=H \textrm{mod}(p)$.

Let $H\in\textrm{Mat}_{n-k\times n}(\zpr)$ be a (full rank) parity-check matrix of an $(n,k)$ free linear block code $\mathcal{C}\subset\zprn$ with $p$-adic expansion $H=\sum_{i=0}^{r-1}H_{i}p^{i}$. For each $i=0,\hdots,r-1$, we construct the matrices $^{(i)}\widehat{H}=p^{i}H$ and $\widehat{H}=(^{(0)}\widehat{H}|^{(1)}\widehat{H}|\cdots|^{(r-1)}\widehat{H})$. Clearly, $^{(0)}\widehat{H}=H$. Consider the $p$-adic expansions of these matrices, $^{(i)}\widehat{H}=\sum_{j=i}^{r-1}\Lambda_{ij}p^{j}$. The following facts are easy to check:
\begin{enumerate}
\item $\Lambda_{00}=H_0$ is full rank.
\item $\Lambda_{ij}=\Lambda_{0j-i}=H_{j-i}$. So $\Lambda_{ii}=H_0$ is full rank.
\item $\Ker(\widehat{H})=\Ker(H)=\mathcal{C}$.
\end{enumerate}
Under these conditions, TLN algorithm can be applied to compute the error $e$ from the equation $\widehat{s}=\widehat{H}e$, where $\widehat{s}=(^{(0)}\widehat{s}|^{(1)}\widehat{s}|\cdots|^{(r-1)}\widehat{s})$ and $^{(i)}\widehat{s}:=p^{i}s$.

Let $e=\sum_{j=0}^{r-1}e_{j}p^{j}$ and $s=\sum_{j=0}^{r-1}s_{j}p^{j}$ be the $p$-adic expansions of $e$ and $s$, respectively. Note that the $p$-adic expansion of $^{(i)}\widehat{s}$ is $^{(i)}\widehat{s}=\sum_{j=i}^{r-1}s_{j-i}p^{j}$. Now, the equation $^{(i)}\widehat{s}=\widehat{H}e ^{(i)}$ leads to the following equation (see \cite[Eq. (12)]{lincodesringsH} and condition 2) above)
\begin{equation}\label{eq:key}
\sum_{j=i}^{r-1}s_{j-i}p^{j}=\sum_{j=i}^{r-1}\sum_{m=0}^{j-i}H_{j-i-m}e_{m} p^{j}
\end{equation}
Assume that $e_0,\hdots,e_{l-1}$ are known.  Then, from Eq. (\ref{eq:key}) it follows that there exists $\delta_l \in\mathbb{F}_{p}$ such that
\begin{equation}\label{eq:key2}
\begin{split}
\Xi(l)&:=\big(\sum_{j=r-l}^{r-1}s_{j-r+l}p^{j}\big)\\
&-\big(\sum_{j=r-l}^{r-2}\sum_{m=0}^{j-r+l}H_{j-r+l-m}e_{m} p^{j}\big)\\
&-\big(\sum_{m=0}^{l-1}H_{l-m-1}e_{m}p^{r-1}\big)=\delta_l p^{r-1}
\end{split}
\end{equation}
holds true. The three summands appearing in $\Xi(l)$ are considered as elements in $\zpr$.
On the other hand, note that the left hand side is known, so $\delta_l$ can be computed. Thus, Eq. (\ref{eq:key}) implies that (see \cite[Eq. (20)]{lincodesringsH})
$$
\delta_l=H_{0}e_{l} 
$$
in $\zp$. At this point we may apply a decoder $\partial$ to recover $e_l$ as long as each error component $e_i$ lies within the correction capacity of $\partial$.

%
%
%
%
%
%
%

The pseudocode of the algorithm is as follows:

\begin{algorithm}[H]
\caption{Adapted TLN algorithm}\label{alg:H}
\begin{algorithmic}
\State \hspace{-0.4cm} \textbf{Parameters:}   $H\in\textrm{Mat}_{n-k\times n}(\zpr)$ full rank parity-check matrix of an $(n,k)$ free linear block code and a decoding algorithm $\partial$ for the linear code given by $H_0:=H \ \textrm{mod}(p)$.
\State \hspace{-0.4cm} \textbf{Input:}  A syndrome $s:=H\widehat{v} \in\mathbb{Z}^{n-k}_{p^{r}}$.
\State \hspace{-0.4cm}  \textbf{Output:} $e=\sum_{j=0}^{r-1}e_{j}p^{j}\in\zprn$.
\State $i=0$
\While{$i \leq r-1$}  
\State  $\delta_i \gets \dfrac{\Xi(i)}{p^{r-1}}$
\State Solve $\delta_i= H_0 e_i$  with $\partial$
\State $i=i+1$
\EndWhile
\State \Return{$e,v=\widehat{v}- e$}
\end{algorithmic}
\end{algorithm}

\subsection{Rosenthal's decoding algorithm for convolutional codes over $\zpr$}

Let $\mathcal{C} \subset \zpr[z]^n$ be an $(n,k)$ observable convolutional code  with PDP and complexity $\delta_{\mathcal{C}}=\delta$, and let $\Sigma=(A,B,C,D)$ be a minimal I/S/O representation over $\zpr$. Rosenthal's decoding algorithm is based on this fact:
a code sequence 
$\left\{c_t= \left(\begin{array}{c}y_t \\ u_t \end{array}\right)\right\}$
must satisfy
\begin{equation}\label{eq:trajectories1}
\begin{split}
&\left( \begin{array}{c}
y_{\tau+l+1}\\
y_{\tau+l+2}\\
\vdots\\
y_{\tau+l+\Theta}
\end{array}\right)-
\left( \begin{array}{ccccc}
D & 0 & \hdots &  & 0\\
CB & D & \ddots & & \vdots\\
CAB & CB & \ddots & \ddots & \\
\vdots & & \ddots & \ddots & 0\\
CA^{\Theta-2}B & CA^{\Theta-3}B & \hdots & CB & D
\end{array}\right)
\left( \begin{array}{c}
u_{\tau+l+1}\\
u_{\tau+l+2}\\
\vdots\\
u_{\tau+l+\Theta}
\end{array}\right)=\\
=&\left( \begin{array}{c}
C\\
CA\\
\vdots\\
CA^{\Theta-1}
\end{array}\right)x_{\tau+l+1}
\end{split}
\end{equation}
and
\begin{equation}\label{eq:trajectories2}
x_{\tau+l+1}-
A^{l+1}x_{\tau}=
\left(
A^{l} B, A^{l-1}B,\hdots,B
\right)
\left( \begin{array}{c}
u_{\tau}\\
u_{\tau+1}\\
\vdots\\
u_{\tau+l}
\end{array}\right)
\end{equation}
for every $\tau,l,\Theta\in\mathbb{N}$.
We will explain now the basis of Rosenthal's decoding algorithm. Let us make the next assumptions for the moment.
\begin{enumerate}
\item $\Phi_{l+1}(\Sigma)$ is surjective and the code $\Ker(\Phi_{l}(\Sigma))$ has distance $d_1$.
\item $\Psi_{\Theta}(\Sigma)$ is injective and the code $\im(\Psi_{\Theta}(\Sigma))$ has distance $d_2$.
\item We have computed the codewords from $t=0$ to $t=\tau-1$ and we know $x_{\tau}$. 
\item The inputs ($u's$) in Eq. (\ref{eq:trajectories1}) has been received correctly.
\item The vector of outputs ($y$'s) in Eq. (\ref{eq:trajectories2}) has weight less than or equal to $[\dfrac{d_2-1}{2}]$.
\item The vector of inputs in Eq. (\ref{eq:trajectories2}) has weight less than or equal to $[\dfrac{d_1-1}{2}]$.
\end{enumerate}
Note that under above assumptions 2), 4) and 5), we may decode $\Psi_{\Theta}(\Sigma)$ (see (\ref{sec:block}), Eq. (\ref{eq:dmin})) and find 
\begin{equation}\label{eq:step1}
\left( \begin{array}{c}
y_{\tau+l+1}\\
y_{\tau+l+2}\\
\vdots\\
y_{\tau+l+\Theta}
\end{array}\right)
\textrm{ and }
x_{\tau+l+1}
\end{equation} 
Now, with $x_{\tau+l+1}$ we can compute the syndrome (see Eq. (\ref{eq:trajectories2}))
$$
\left(
A^{l} B, A^{l-1}B,\hdots,B
\right)
\left( \begin{array}{c}
\widehat{u}_{\tau}-u_{\tau}\\
\widehat{u}_{\tau+1}-u_{\tau+1}\\
\vdots\\
\widehat{u}_{\tau+1}-u_{\tau+l}
\end{array}\right)
$$
and, under the assumptions 1) and 6), we can decode $\Phi_{l+1}(\Sigma)$ and find
$$\left( \begin{array}{c}
u_{\tau}\\
u_{\tau+1}\\
\vdots\\
u_{\tau+l}
\end{array}\right)$$
since 
$$
\left( \begin{array}{c}
\widehat{u}_{\tau}-u_{\tau}\\
\widehat{u}_{\tau+1}-u_{\tau+1}\\
\vdots\\
\widehat{u}_{\tau+1}-u_{\tau+l}
\end{array}\right)
$$
is the unique coset leader in its coset (see \ref{sec:block}).
So far, we have correctly computed the input vectors ($u$'s) from $t=\tau$ to $t=\tau+l+\Theta$ and the output vectors ($y$'s) from $t=\tau+l+1$ to $\tau+l+\Theta$, and we can use the equations (\ref{eqslinearsys}) to compute the remaining outputs from $t=\tau$ to $t=\tau+l$ so the decoding is complete.

In general, above assumptions 5) and 6) might not hold so the decoding process may give no output or may give an error. However, under certain mild assumptions the algorithm ensures that we get the correct output. In fact, \cite[Theorem 3.4]{decodingros} can be stated exactly in the same terms for observable convolutional codes with the PDP over $\zpr$.
\begin{theorem}\label{th:theorem}
Let $\mathcal{C}\subset\zprz^{n}$ be an $(n,k)$  observable PDP convolutional code, and $\Sigma=(A,B,C,D)$ a minimal I/S/O representation of $\mathcal{C}$. 
Let $T>\Theta$ be natural numbers such that 
\begin{enumerate}
\item $A$ is invertible.
\item $\Phi_{T}(\Sigma)$ is surjective and the code $\Ker(\Phi_{l}(\Sigma))$ has distance $d_1$.
\item $\Psi_{\Theta}(\Sigma)$ is injective and the code $\im(\Psi_{\Theta}(\Sigma))$ has distance $d_2$.
\end{enumerate}
Then, if we receive a message
$$
\left\{
\left(
\begin{array}{c}
\widehat{y}_{t}\\
\widehat{u}_t
\end{array}
\right)
\right\}_{t\geq 0}=
\left\{
\left(
\begin{array}{c}
y_{t}+f_{t}\\
u_{t}+e_{t}
\end{array}
\right)
\right\}_{t\geq 0}
$$
such that, for any $\tau\geq 0$, 
 the weight of the errors in the time-window $[\tau,\tau+T-1]$
is bounded by 
\begin{equation}\label{eq:lambda}
\lambda:=\textrm{min}\left(
\Big\lfloor\dfrac{d_1 -1}{2}\Big\rfloor,\Big\lfloor\dfrac{T}{2\Theta}\Big\rfloor
\right),
\end{equation}
it is possible to uniquely compute the transmitted sequence.
\end{theorem}
\begin{proof}
The proof follows exactly the same steps as in \cite{decodingros}, that is, we decode iteratively as stated above for $l=T-\Theta$ (see \textbf{Algorithm} \ref{alg:dec}). If it is not possible to decode in Eq. (\ref{eq:trajectories1}), or the we find that the weight of the computed errors from $t=\tau$ to $t=\tau+T-\Theta$ is larger than $\lambda$, we repeat again the procedure with $l=T-2\Theta$ and substitute $\lambda$ by $\lambda-1$. We proceed in this way until we can decode successfully. As in the field case we can not be sure that the calculation of $x_{\tau+l}$ (for $l=T-h\Theta$ if we needed $h$ iterations) is correct. However, \cite[Lemma 3.3]{decodingros} also holds in the ring case, so the conclusion that the state $x_{\tau+\Theta}$ is correct holds as well in our situation. The proof of \cite[Lemma 3.3]{decodingros} in the ring case is the same as in the field case. Note that, in our situation, \cite[Eq. (3.6)]{decodingros} implies that the distance between the vector of true information words modulo $p$ and the vector of computed information words modulo $p$ is at least $d_1$ because the corresponding linear block code is splitting so its distance remains the same after restricting modulo $p$. Since the distance between vectors is at least the distance between the vectors restricted modulo $p$, we are done because the rest of the argument remains valid in our case.
\end{proof}
\begin{remark}[The splitting condition]
It is important to note that, in the proof of the above theorem, we did not require the code $\operatorname{Im}(\Psi_{\Theta}(\Sigma))$ to be splitting. Indeed, the theoretical success of Rosenthal’s decoding algorithm relies only on the assumption that the associated linear block codes can, in principle, be decoded—regardless of the existence of an explicit decoding procedure. However, in order to implement the algorithm in practice, it becomes necessary to fix specific decoding algorithms for the two associated linear block codes. Under our assumptions, one can always choose parameters $T$ and $\Theta$ such that both $\operatorname{Ker}(\Phi_T(\Sigma))$ and $\operatorname{Im}(\Psi_{\Theta}(\Sigma))$ are splitting codes (in fact the first one is always splitting). Since efficient decoding algorithms are available for splitting linear codes over $\mathbb{Z}_{p^r}$, we formulate the final result by explicitly exploiting this structure.
\end{remark}

\begin{algorithm}[h]
\caption{Combined algorithm}\label{alg:dec}
\begin{algorithmic}
\State \hspace{-0.4cm} \textbf{Parameters:}  Minimal I/S/O $\Sigma=(A,B,C,D)$ of an observable convolutional code with PDP, natural numbers $T>\Theta$ satisfying assumptions of Theorem \ref{th:theorem}, $\lambda$ as in Eq. (\ref{eq:lambda}).
\State \hspace{-0.4cm} \textbf{Input:} \\
$\bullet$  
Window of length $T$ of received pairs $\left\{(\hat{y}_t, \hat{u}_t)\right\}_{t=\tau}^{\tau+T-1}$\\
$\bullet$ A natural number $h$ with $T\geq h\Theta$.
\State \hspace{-0.4cm}  \textbf{Output:} \\
$\bullet$  Estimation of the sequence $\left\{(\tilde{y}_t, \tilde{u}_t)\right\}_{t=\tau}^{\tau+\Theta-1}$ and state $x_{\tau+\Theta}$\\
$\bullet$ Or: an ERROR  message together with $h\gets h+1$.\\
\State $\lambda\gets\lambda-h$
\State $\Theta\gets h\Theta$
\State $l\gets T-h\Theta$
\State \textbf{Step 1: Decode the Output Sequence}
\If {Decoding of Eq. (8) using $\Psi_\Theta(\Sigma)$ succeeds}
    \State We get $x_{\tau+l+1}$ and may compute $y_{\tau+l+1}, \dots, y_{\tau+l+\Theta}$

    \State \textbf{Step 2: Syndrome Decoding for Input}
    \State Compute syndrome from Eq. (9) using $x_{\tau+l+1}$
    \If {Decoding using $\Phi_{l+1}(\Sigma)$ succeeds}
        \State Compute $(u_{\tau},y_{\tau}), \dots, (u_{\tau+l},y_{\tau+l})$ using Eq. (4)
    \Else
        \State \textbf{return} ERROR, $h \gets h + 1$
    \EndIf
\Else
    \State \textbf{return} ERROR, $h \gets h + 1$
\EndIf

\State \textbf{Step 3: Verify Error Weight}
\If {Total weight of errors $\sum_{\tau}^{\tau+T-\Theta}(\omega_{H}(f_{t})+\omega_{H}(e_{t}))>\lambda$}
    \State \textbf{return} ERROR, $h \gets h + 1$
\Else
    \State \textbf{return} Estimated 
$\left\{(\tilde{y}_t, \tilde{u}_t)\right\}_{t=\tau}^{\tau+\Theta-1}$
and state $x_{\tau+\Theta}$
\EndIf
\end{algorithmic}
\end{algorithm}

\begin{corollary}
Let $\mathcal{C}\subset\zprz^{n}$ be an $(n,k)$  observable PDP convolutional code, and $\Sigma=(A,B,C,D)$ a minimal I/S/O representation of $\mathcal{C}$.
Let $T>\Theta$ be natural numbers satisfying the same conditions as in Theorem \ref{th:theorem} and additionally such that the linear code $\im(\Psi_{\Theta}(\Sigma))$ is splitting. 
Let $\partial_{gen}$, $\partial_{par}$ be decoding algorithms for the linear codes (over $\zp$) determined by $\im(\Psi_{\Theta}(\Sigma_0))$ (the first coefficient in the $p$-adic expansion of $\Psi_{\Theta}(\Sigma)$), and $\Ker(\Phi_{T}(\Sigma_0))$ (the first coefficient of the $p$-adic expansion of $\Phi_{T}(\Sigma)$) and let $t_{gen}, t_{par}$ be the correction capabilities of $\partial_{gen}$, $\partial_{par}$, respectively.
Then, if we receive a message
$$
\left\{
\left(
\begin{array}{c}
\widehat{y}_{t}\\
\widehat{u}_t
\end{array}
\right)
\right\}_{t\geq 0}=
\left\{
\left(
\begin{array}{c}
y_{t}+f_{t}\\
u_{t}+e_{t}
\end{array}
\right)
\right\}_{t\geq 0}
$$
such that, for any $\tau\geq 0$, 
\begin{enumerate}
\item the weight of the errors 
$
\left\{
\left(
\begin{array}{c}
f_{t}\\
e_{t}
\end{array}
\right)
\right\}_{t= \tau}^{\tau+T-1}$
is bounded by $\lambda:=\textrm{min}\left(
\Big\lfloor\dfrac{d_1 -1}{2}\Big\rfloor,\Big\lfloor\dfrac{T}{2\Theta}\Big\rfloor
\right)$.
\item the sequence of the $i$th coefficients in the $p$-adic expansion of the vectors of the error sequence $(f_{\tau},\hdots,f_{\tau+T-1})$ has weight less or equal to $t_{gen}$,
\item the sequence of the $i$th coefficients in the $p$-adic expansion of the vectors of the error sequence $(e_{\tau},\hdots,e_{\tau+T-1})$ has weight less or equal to $t_{par}$,
\end{enumerate}
it is possible to uniquely compute the transmitted sequence.
\end{corollary}
\begin{proof}
It follows from Theorem \ref{th:theorem}, the splitting assumption and applying \textbf{Algorithm} \ref{alg:G} and \textbf{Algorithm} \ref{alg:H} in the decoding steps of \textbf{Algorithm} \ref{alg:dec}.
\end{proof}

%
%
%
%

\section{Performance and complexity of the algorithm}
Let $\mathcal{C}\subset\zprz^{n}$ be an $(n,k)$  observable PDP convolutional code, and $\Sigma=(A,B,C,D)$ a minimal I/S/O representation of $\mathcal{C}$.
Let $T>\Theta$ be natural numbers as in Theorem \ref{th:theorem}.

Regarding the performance of the algorithm, it is important to note that, essentially, it depends on the detection and correction capacity of the convolutional code modulo $p$. More precisely, in case 
$$\lambda=\Big\lfloor\dfrac{d_1 -1}{2}\Big\rfloor,$$
the theoretical performance of the algorithm depends on the minimum distance of the linear code whose parity-check matrix is given by $\Phi_{T}(\Sigma)$. Since this code is free, its minimum distance coincides with  the linear code over $\zp$ whose {parity-check matrix} is given by $\Phi_{T}(\Sigma)\textrm{mod} \ p$. This, together with \cite[Remark 7]{isoringmunoz}, allow to conclude (as in \cite[Remark 3.6]{decodingros}) that decoding is theoretically possible if no more than 
$$\dfrac{d_{\textrm{free}}(\mathcal{C} \ \textrm{mod} \ p)-1}{2}$$
errors occur in any time interval of length $T$.

Finally, regarding the time complexity of the algorithm, we have:

\begin{theorem}\label{thm:time-complexity}
The worst-case time complexity of Algorithm~3 satisfies
\begin{align*}
C_{\mathrm{total}}
=
O\!\left(\frac{T^2}{\Theta}\,\big(n\delta^2+n\delta k+n k+n\delta\big)\right)
+
\sum_{h=1}^{\lfloor T/\Theta\rfloor}\big(f(h\Theta)+g(T-h\Theta+1)\big).
\end{align*}
\end{theorem}

\begin{proof}
Fix an iteration and set $\Theta_h := h\Theta$, $l_h := T-\Theta_h$,
$H := \Big\lfloor \frac{T}{\Theta}\Big\rfloor$.
We bound each contribution to the cost of Algorithm~3 in that iteration.
\begin{enumerate}
\item $C_{\Psi}(h)=O(\textrm{Construction of $\Psi_{\Theta_h}(\Sigma)$})$.
By definition,
$\Psi_{\ell}(\Sigma)=\big(C,\;CA,\;\ldots,\;CA^{\ell-1}\big)^{\top}$,
which has size $\ell(n-k)\times\delta$. Construct it iteratively by setting
$M_0=C$ and $M_i=M_{i-1}A$ for $i=1,\ldots,\ell-1$.
Each multiplication $M_{i-1}A$ multiplies an $(n-k)\times\delta$ matrix by a
$\delta\times\delta$ matrix, hence costs $O((n-k)\delta^2)$ operations.
With $\ell=\Theta_h$ iterations, this yields
\[
C_{\Psi}(h)=O\big(\Theta_h\,(n-k)\,\delta^2\big).
\]

\item $C_{\Phi}(h)=O(\textrm{Construction of $\Phi_{l_h+1}(\Sigma)$)}$.
By definition,
$\Phi_{\ell}(\Sigma)=\big[A^{\ell-1}B,\;A^{\ell-2}B,\;\ldots,\;B\big]$,
which has size $\delta\times \ell k$. Construct the block columns iteratively by
$N_0=B$ and $N_i=AN_{i-1}$ for $i=1,\ldots,\ell-1$.
Each multiplication $AN_{i-1}$ multiplies a $\delta\times\delta$ matrix by a
$\delta\times k$ matrix, hence costs $O(\delta^2 k)$ operations.
With $\ell=l_h+1$ blocks, we obtain
\[
C_{\Phi}(h)=O\big((l_h+1)\,\delta^2\,k\big).
\]

\item $C_{\mathrm{aux}}(h)=O(\textrm{Computations})$.
In iteration $h$, the algorithm performs linear-time
operations such as state/output propagation and recomputation via
$x_{t+1}=Ax_t+Bu_t$ and $y_t=Cx_t+Du_t$ for $O(\Theta_h+l_h+1)$ time steps.
Each step costs at most
$O(\delta^2)+O(\delta k)+O((n-k)\delta)+O((n-k)k)
=
O(\delta^2+\delta k+(n-k)\delta+(n-k)k)$,
so
\[
C_{\mathrm{aux}}(h)
=
O\Big((\Theta_h+l_h+1)\,(\delta^2+\delta k+(n-k)\delta+(n-k)k)\Big).
\]
Any additional matrix--vector products, e.g.\ multiplying $\Phi_{l_h+1}$ by a vector
at cost $O(\delta(l_h+1)k)$, are dominated by the above term and/or by $C_{\Phi}(h)$.

\item $C_{\mathrm{check}}(h)=O(\textrm{Error-weight verification})$.
If the weight of each pair $(f_t,e_t)$ is computed by scanning $n$ components,
then over a window of length $T$ the verification costs $O(Tn)$, i.e.,
$C_{\mathrm{check}}(h)=O(Tn)$.
\end{enumerate}
Summing the per-iteration bounds for $h=1,\ldots,H$, and adding the decoding costs $f(\Theta_h)$ and $g(l_h+1)$, yields the expression 
\begin{equation}\label{eq:complexity-inter}
C_{\mathrm{total}}
=
\sum_{h=1}^{H}\Big(
C_{\Psi}(h) + f(\Theta_h)
+ C_{\Phi}(h) + g(l_h+1)
+ C_{\mathrm{aux}}(h)
+ C_{\mathrm{check}}(h)
\Big),
\end{equation}
Now, from $\Theta_h=h\Theta$ and $H=\lfloor T/\Theta\rfloor$ we obtain
\[
\sum_{h=1}^{H}\Theta_h
=
\Theta\sum_{h=1}^{H}h
=
\Theta\frac{H(H+1)}{2}
=
O\!\left(\frac{T^2}{\Theta}\right).
\]
Moreover, since $l_h=T-h\Theta$, it holds
\[
\sum_{h=1}^{H}(l_h+1)
=
\sum_{h=1}^{H}(T-h\Theta+1)
=
H(T+1)-\Theta\frac{H(H+1)}{2}
=
O\!\left(\frac{T^2}{\Theta}\right),
\]
and trivially $\sum_{h=1}^{H}Tn=HTn=O\!\left(\frac{T^2 n}{\Theta}\right)$.
Plugging these estimates into (\ref{eq:complexity-inter}) and keeping the sums of $f(h\Theta)$ and $g(T-h\Theta+1)$ explicit yields
\begin{align*}
C_{\mathrm{total}}
&=
O\!\left(
\frac{T^2}{\Theta}\Big(
(n-k)\delta^2+\delta^2k+\delta^2+\delta k+(n-k)\delta+(n-k)k+n
\Big)
\right)
\\
&\quad
+\sum_{h=1}^{\lfloor T/\Theta\rfloor}\big(f(h\Theta)+g(T-h\Theta+1)\big).
\end{align*}
Finally, since $n\geq k\geq 1$ and $\delta\geq 0$, it holds
$(n-k)\delta^2+\delta^2k+\delta^2+\delta k+(n-k)\delta+(n-k)k+n
=O(n\delta^2+n\delta k+nk+n\delta)$. Substituting this into the previous expression, we get the result.
\end{proof}

\section{Conclusion and future work}


In this article, we have extended Rosenthal's decoding algorithm based on I/S/O representations from the field case to the setting of convolutional codes over the finite ring \( \mathbb{Z}_{p^r} \). The proposed decoding method combines system-theoretic techniques with block code decoders tailored for ring structures, and it is capable of handling noisy transmission channels through a modular and structured approach.

In future work, we aim to adapt this framework to handle erasure channels, where I/S/O-based decoding strategies have already shown advantages in the field case. Specifically, our goal is to construct a decoding procedure for convolutional codes over finite rings in the presence of erasures and to evaluate its performance relative to recent contributions such as \cite{pinto-decoding-g}, which target general codes over erasure channels. This direction can potentially enlarge the applicability of the system-theoretic decoding paradigm and further bridge the gap between abstract algebraic models and practical error-control strategies.


%




\bibliography{biblio}
\bibliographystyle{plain}

\end{document}